\documentclass[onefignum,onetabnum]{siamart190516}

\usepackage{lipsum}
\usepackage{amsfonts}
\usepackage{graphicx}
\usepackage{epstopdf}
\usepackage{amsmath,amssymb,epsfig,graphics,psfrag,latexsym,amsmath,listings,url,tabularx,gensymb}
\usepackage{stmaryrd}
\usepackage[utf8]{inputenc}
\usepackage{multicol}
\usepackage{bbm}
\usepackage{color}
\usepackage{caption}
\usepackage{blkarray,bigstrut,xparse} 
\usepackage{subcaption}
\usepackage{graphicx,xcolor} 
\usepackage[framemethod=tikz]{mdframed}
\Crefname{ALC@unique}{Line}{Lines} 
\usepackage{algorithmic}
\graphicspath{ {../Figures2/} }
\ifpdf
  \DeclareGraphicsExtensions{.eps,.pdf,.png,.jpg}
\else
  \DeclareGraphicsExtensions{.eps}
\fi

\usepackage{siunitx} 
\usepackage{booktabs}

\usepackage{rotating}

\newsiamremark{remark}{Remark}
\newsiamremark{hypothesis}{Hypothesis}
\crefname{hypothesis}{Hypothesis}{Hypotheses}
\newsiamthm{claim}{Claim}

\newcommand{\Prob}{\mathbb{P}}
\newcommand{\curlyP}{\mathcal{P}}

\newcommand{\curlyE}{\mathcal{E}}

\newcommand{\curlyQ}{\mathcal{Q}}

\newcommand{\curlyX}{\mathcal{X}}
\newcommand{\curlyY}{\mathcal{Y}}
\newcommand{\Comments}{1}
\newcommand{\mynote}[2]{\ifnum\Comments=1\textcolor{#1}{#2}\fi}
\newcommand{\mytodo}[2]{\ifnum\Comments=1%
  \todo[linecolor=#1!80!black,backgroundcolor=#1,bordercolor=#1!80!black]{#2}\fi}

\headers{On Contamination of Symbolic Datasets}{A. Pearson and M. E. Lladser}

\title{On Contamination of Symbolic Datasets\thanks{Submitted to the editors DATE.
\funding{This work was partially supported by the NSF Graduate Research Fellowship Program grant No. 2016198773 (Pearson), and the NSF IGERT grant No. 1144807.}}}

\author{Antony Pearson\thanks{Department of Applied Mathematics, University of Colorado, Boulder, CO.}
\and Manuel E. Lladser\thanks{Department of Applied Mathematics, University of Colorado, Boulder, CO
  (email{manuel.lladser@colorado.edu}).}
}
\usepackage{amsopn}

\ifpdf
\hypersetup{
pdftitle={On Contamination of Symbolic Datasets},
pdfauthor={A. Pearson and M. E. Lladser}
}
\fi

\begin{document}

\maketitle

\begin{abstract}
Data taking values on discrete sample spaces are the embodiment of modern biological research. ``Omics'' experiments produce millions of symbolic outcomes in the form of reads (i.e., DNA sequences of a few dozens to a few hundred nucleotides). Unfortunately, these intrinsically non-numerical datasets are often highly contaminated, and the possible sources of contamination are usually poorly characterized. This contrasts with numerical datasets where Gaussian-type noise is often well-justified. To overcome this hurdle, we introduce the notion of latent weight, which measures the largest expected fraction of samples from a contaminated probabilistic source that conform to a model in a well-structured class of desired models. We examine various properties of latent weights, which we specialize to the class of exchangeable probability distributions. As proof of concept, we analyze DNA methylation data from the 22 human autosome pairs.  Contrary to what it is usually assumed, we provide strong evidence that highly specific methylation patterns are overrepresented at some genomic locations when contamination is taken into account.
\end{abstract}



\section{Introduction}
Symbolic data is the epitome modern biological datasets due to the advent of high-throughput sequencing assays. These assays generate millions of comparatively short DNA sequences of a few dozen to a few hundred nucleotides and allow scientists to assess various microscopic processes such as investigating the relative abundance of unculturable organisms in an environment~\cite{LlaGouRee11,HamLla12}, or pinpointing the location of the enzymes that are actively transcribing DNA into RNA along a genome~\cite{LlaAzoAllDow17}---among many other possibilities~\cite{RNA-seq,GRO-seq,ChIP-seq,methylation_assays}. Unfortunately, these datasets are often very noisy, and it is often unclear how to describe the noise because the possible sources of data corruption can be so intricate that there is little motivation for any specific and let alone universal representation of it. In contrast, Gaussian errors are often well-justified with continuous numerical data.

There is a rich history of using mixtures to describe deviations from idealized continuous models~\cite{newcomb1886,tukey1960,huber1964,huber1965}. Such mixtures are usually of the form: $P = (1-\epsilon)\cdot Q+\epsilon\cdot R$, where $P$ is the probabilistic source producing the data, $Q$ is a Gaussian distribution, and $R$ is some contaminating probability distribution from which an observation is drawn with some small probability $\epsilon$. When the mixture is unspecified, and $P$ ought to be estimated from data, a highly specific structure for $R$ is usually needed, e.g. a Gaussian with known mean, to make $\epsilon$ and $Q$ identifiable~\cite{contGaussMix}.

In this manuscript, we address the problem of assessing contamination (or its opposite, ``purity,'') in symbolic datasets---which are intrisically discrete. Like prior work on continuous data, we model contamination as a mixture. Unlike previous lines of work, we treat contamination as incidental; in particular, we do not commit to any prespecified form for it. To do so, we introduce the notion of latent weight with respect to a given structured class of probabilistic models (e.g., exchangeable probability distributions, which are the main focus of this manuscript). Broadly speaking, this is the largest weight a model in the structured class can have as a component of the source producing the data. In particular, it describes the largest expected fraction of samples from a contaminated random sample which conform to a probabilistic model in the structured class. 

We argue that latent weights are always identifiable, and allow one to represent unstructured probabilistic models as mixtures with a well-structured component; in particular, when this component carries a substantial latent weight, most samples from the mixture can be attributed to it. Latent weights offer therefore a measure of the truthiness of a hypothesis, which may not be strictly true.

To fix ideas, consider binary random variables $X$ and $Y$ with joint probability distribution given by the matrix
\begin{equation}
P :=
\begin{blockarray}{ccc}
\hbox{\tiny $Y=0$} & \hbox{\tiny $Y=1$} \\
\begin{block}{(cc)c}
1/10 & 3/10 &  \hbox{\tiny $X=0$} \\
1/10 & 1/2 & \hbox{\tiny $X=1$} \\
\end{block}
\end{blockarray}.
\label{ide:introex}
\end{equation}
Since $X\sim Bernoulli(3/5)$ and $Y\sim Bernoulli(4/5)$, it is straightforward to check that $X$ and $Y$ are not independent. Nevertheless, we can ask whether or not $P$ can be represented as a mixture with a component with independent marginals. This is related to determining if the latent weight of $P$ with respect to the class $\curlyQ$ of product measures of the form $(\mu\otimes\nu)$, with $\mu$ and $\nu$ probability models supported on $\{0,1\}$, is positive or not (see Definition~\ref{def:main}). It turns out the largest weight one can give to the model $Bernoulli(3/5)\otimes Bernoulli(4/5)$ formed using the marginals of $(X,Y)$ is $5/6\approx 83\%$. Indeed:
\[P = \frac{5}{6}\cdot Bernoulli(3/5)\otimes Bernoulli(4/5)+\frac{1}{6} \cdot\left(\begin{array}{cc}
1/5 &  1/5\\
0 &  3/5
\end{array}\right).
\]
The latent weight of $P$ w.r.t. $Q$ must be therefore at least $5/6$. In fact, a simple calculation reveals that
\[P=\frac{24}{25}\cdot Bernoulli(5/8)\otimes Bernoulli(5/6) + \frac{1}{25}\cdot\left(\begin{array}{cc}
1 &  0\\
0 &  0
\end{array}\right),\]
i.e., the latent weight of $P$ with respect to $\curlyQ$ is at least $24/25=96\%$. Further analysis would show that it is precisely this; in particular, up to a hidden event with $96\%$ probability, $X$ and $Y$ behave independently. This finding is noteworthy for many reasons. 

On one hand, perhaps unexpectedly, the marginal distributions of $X$ and $Y$ are not associated with the latent weight of $P$ with respect to $\curlyQ$. On the other hand, if we derived our beliefs about the independence of $X$ and $Y$ from a large sample from $P$, e.g. using a chi-squared test of independence, at moderate significance levels we would typically reject the hypothesis that $X$ and $Y$ are independent, missing that most of the time this is not the case. In this regard, latent weights could be used as a proxy for the truthiness of a hypothesis, without being tied to the absolutes of truth or falsity in classical hypothesis testing approaches. Finally, if $P$ were estimated from a large but corrupted sample, as is usually the case with modern biological datasets, the high weight of $Bernoulli(5/8)\otimes Bernoulli(5/6)$ as a component of $P$ would suggest modeling $X$ and $Y$ as independent---as opposed to a more complicated model with a spurious correlation induced by a seemingly 4\% contamination

\noindent\textbf{Paper organization.} Since our primary interest is on symbolic data, we restrict ourselves to the setting of finite sample spaces. From a technical point of view, this helps to develop the theory but without trivializing it. Section~\ref{sec:general} introduces and analyzes the notion of latent weight associated with a general class of probability models. Sections~\ref{sec:exchangeableweight}-\ref{sec:proofofconcept} are exclusively devoted to the class of exchangeable probability models. Section~\ref{sec:upperbounds} introduces some upper-bounds that may be useful for assessing an exchangeable latent weight when the sample space is too large relative to the number of observations from the model. Section~\ref{sec:estimation} develops the statistical machinery necessary for inference of exchangeable latent weights when the probabilistic source (producing the data) is observed only indirectly through a random sample. Finally, Section~\ref{sec:proofofconcept} demonstrates the use of latent weights to assess the exchangeability of DNA methylation, a near-universal assumption in epigenomic analyses.

\section{Latent weights}\label{sec:general}

In what follows, $\curlyP$ \textbf{denotes the set of all probability measures over certain finite non-empty sample space $\Omega$.} As such, $\curlyP$ is compact under any norm induced metric; in particular, the total variation norm. Recall that for $P_1,P_2\in\curlyP$, this norm is defined as~\cite{lindvallCoupling}:
\[\|P_1-P_2\|:=\max_{A\subset\Omega}|P_1(A)-P_2(A)|=\frac{1}{2}\sum_{x\in\Omega}|P_1(x)-P_2(x)|.\]

In what remains of this manuscript, $\mathbf{\curlyQ\subset\curlyP}$ \textbf{denotes a closed non-empty subset of probability measures.} In particular, $\curlyQ$ is a compact subset of $\curlyP$. Instances like this include, for example, singletons, as well as models with independent marginals (when $\Omega$ is a product space), among various other possibilities.

\begin{definition}
\label{def:main}
Let $P\in\curlyP$. We define the \textit{(latent) weight of $\curlyQ$ in $P$} as the coefficient $\lambda_\curlyQ(P):=\sup\{\lambda\hbox{ such that }P\ge\lambda\cdot Q\hbox{ for some }Q\in\curlyQ\}$, where $P\ge\lambda\cdot Q$ means that $P(\omega)\ge\lambda\cdot Q(\omega)$, for all $\omega\in\Omega$.
\end{definition}

This definition resembles that given in~\cite{CheLla10,LlaChe14} but for the very different purpose of representing a long-lasting Markov chain by shorter-lived independent chains.

Clearly, $0\le\lambda_\curlyQ(P)\le 1$. Latent weights have various other properties which we now state.

\begin{theorem}
$\lambda_\curlyQ(P)=1$ if and only if $P\in\curlyQ$.
\label{thm:lambda1closure}
\end{theorem}
\begin{proof}
If $\lambda_\curlyQ(P)=1$ then there exists a sequence of real numbers $(\lambda_n)_{n\ge1}$ such that $1-1/n<\lambda_n\le1$, and $P\ge\lambda_n\cdot Q_n$ for some $Q_n\in\curlyQ$. Without loss of generality assume that $\lambda_n<1$. In particular, if we define $R_n:=(P-\lambda_n\cdot Q_n)/(1-\lambda_n)$ then $R_n\in\curlyP$ and $P=\lambda_n\cdot Q_n+(1-\lambda_n)\cdot R_n$. As a result: $(P-Q_n)=(1-\lambda_n)\cdot(R_n-Q_n)$, which implies that $\|P-Q_n\|=(1-\lambda_n)\cdot\|R_n-Q_n\|<1/n$. Hence $P\in\curlyQ$ because $\curlyQ$ is closed. The converse is immediate because $\lambda_\curlyQ(Q)=1$ for all $Q\in\curlyQ$.
\end{proof}

\begin{theorem}
\begin{equation}
\lambda_\curlyQ(P) =\sup\limits_{Q\in\curlyQ}\,\,\min\limits_{\omega\in\Omega}\frac{P(\omega)}{Q(\omega)},
\label{ide:supinformula}
\end{equation} 
where any division by zero (including zero-over-zero) is to be interpreted as $+\infty$, and this supremum is achieved; in particular, there is $Q\in\curlyQ$ such that $P\ge\lambda_\curlyQ(P)\cdot Q$ and there is $R\in\curlyP$ such that \begin{align}
P=\lambda_\curlyQ(P)\cdot Q+(1-\lambda_\curlyQ(P))\cdot R.
\label{ide:genmixrep}
\end{align}
Further, if $\curlyQ$ is convex and $\lambda_\curlyQ(P)<1$ then  $\lambda_\curlyQ(R)=0$; in particular, $R\notin\curlyQ$.
\label{thm:supminformula}
\end{theorem}
\begin{proof}
Fix $P\in\curlyP$, and define $\lambda^*:=\sup_{Q\in\curlyQ}\min_{\omega\in\Omega}P(\omega)/Q(\omega)$. For each $Q\in\curlyQ$, note that $P\ge\lambda\cdot Q$ if and only if $\min_{\omega\in\Omega}P(\omega)/Q(\omega)\ge\lambda$; in particular, $\lambda^*\ge\lambda_\curlyQ(P)$. Define $\epsilon:=\lambda^*-\lambda_\curlyQ(P)$. If $\epsilon>0$ then, from the definition of $\lambda^*$, there would be $Q\in\curlyQ$ such that $\min_{\omega\in\Omega}P(\omega)/Q(\omega)\ge\lambda^*-\epsilon/2$. In particular, for all $\omega\in\Omega$, $P(\omega)\ge(\lambda^*-\epsilon/2)Q(\omega)$, hence $\lambda_\curlyQ(P)\ge\lambda^*-\epsilon/2>\lambda_\curlyQ(P)$, a contradiction. Therefore $\epsilon=0$ i.e. $\lambda^*=\lambda_\curlyQ(P)$.

Next, note that for each $Q\in\curlyQ$:
\begin{equation}
\min_{\omega\in\Omega}\frac{P(\omega)}{Q(\omega)}=\min_{\omega:\,Q(\omega)>0}\frac{P(\omega)}{Q(\omega)}.
\label{ide:lambdaQPmin1}
\end{equation}
Further, if $\lim_{n\to\infty}Q_n=Q$ in total variation distance then $\lim_{n\to\infty}Q_n(\omega)=Q(\omega)$ uniformly for all $\omega\in\Omega$. Hence, the transformation $Q\in\curlyQ\rightarrow\min_{\omega\in\Omega}P(\omega)/Q(\omega)$ from $\curlyQ$ to $\mathbb{R}$ is continuous; in particular, because $\curlyQ$ is compact, the supremum in equation~(\ref{ide:supinformula}) is achieved, and there is $Q\in\curlyQ$ such that $P\ge\lambda_\curlyQ(P) \cdot Q$. For brevity, define $\lambda:=\lambda_\curlyQ(P)$. Following an argument similar to the proof of Theorem~\ref{thm:lambda1closure}, it follows that there is $R\in\curlyP$ such that $P=\lambda\cdot Q+(1-\lambda)\cdot R$. Likewise, if $\alpha:=\lambda_\curlyQ(R)$ then there is $Q'\in\curlyQ$ and $R'\in\curlyP$ such that $R=\alpha\cdot Q'+(1-\alpha)\cdot R'$. As a result, if $\curlyQ$ is convex then $P\ge\lambda\cdot Q+(1-\lambda)\alpha\cdot Q'=(\lambda+(1-\lambda)\alpha)\cdot Q''$ for some $Q''\in\curlyQ$. But then $\lambda\ge\lambda+(1-\lambda)\alpha$, which implies that $(1-\lambda)\alpha=0$. Thus, if $\lambda<1$ then $\alpha=0$ as claimed.
\end{proof}

We emphasize that the probability measure $Q$ in Theorem~\ref{thm:supminformula} is not necessarily unique. For instance, if $\curlyQ$ is the space of probability measures over $\Omega=\{0,1\}^2$ with i.i.d. marginals, and $P$ is the uniform probability measure over the set $\{(0,0),(1,1)\}$ then $\lambda_\curlyQ(P)=1/2$, and $\delta_{(0,0)}$ and $\delta_{(1,1)}$, the point masses at $(0,0)$ and $(1,1)$, respectively, both achieve the supremum in equation~(\ref{ide:supinformula}). Likewise, since $P=\delta_{(0,0)}/2+\delta_{(1,1)}/2$, this counterexample shows the importance of convexity to guarantee that the probability measure $R$ in equation~(\ref{ide:genmixrep}) does not belong to $\curlyQ$.

The identity in equation~(\ref{ide:genmixrep}) implies that any probability model $P$ over $\Omega$ admits a mixture representation with a component in $\curlyQ$ of weight $\lambda_\curlyQ(P)$. (This motivates the terminology of ``latent weight.'') The weight of $\curlyQ$ in $P$ may be interpreted therefore as the largest expected fraction of observations from $P$ that can be attributed to a single model in $\curlyQ$.

We also note that a large latent weight is indicative of closeness in total variation distance to $\curlyQ$. In fact, if $Q$ and $R$ are as in equation~(\ref{ide:genmixrep}) then (see proof of Theorem~\ref{thm:lambda1closure}): $\|P-Q\|=(1-\lambda_\curlyQ(P))\cdot\|R-Q\|$. So, if $\lambda_\curlyQ(P)$ is close to 1 then $P$ is close to $\curlyQ$ in total variation distance. The converse is not necessarily true, however. For instance, consider $\Omega=\{0,1\}^d$, with finite $d>1$, and let $0^d$ and $1^d$ denote the sequences of $d$ zeros and ones, respectively. Define $P(x):=(2^d-2)^{-1}$ for $x\in\Omega\setminus\{0^d,1^d\}$. In particular, $P(0^d)=P(1^d)=0$, and $\|P-\hbox{Uniform}(\{0,1\}^d)\|=2^{2-d}$, i.e. $P$ is very close in total variation distance to $\curlyQ$. Nevetheless, the latent weight of $P$ with respect to the class $\curlyQ$ of i.i.d. distributions over $\{0,1\}^d$ is $0$. Indeed, for all $\lambda>0$ and $Q:=\otimes_{k=1}^{d}Bernoulli(q)$, with $0\le q\le1$, it is not possible to have $P\ge\lambda\cdot Q$. Otherwise, because $Q(0^d)=(1-q)^d$ and $Q(1^d)=q^d$, it would follow that $q=1$ and $q=0$ simultaneously.

\section{Exchangeable weights}\label{sec:exchangeableweight}

\textbf{The remaining of the manuscript focuses on the class of exchangeable probabilistic models.} Recall that a finite sequence of random variables $X_1,\ldots,X_d$ is called exchangeable if for any permutation $\sigma$ of $(1,\ldots,d)$ the random vector $(X_{\sigma(1)},\ldots,X_{\sigma(d)})$ has the same distribution as $(X_1,\ldots,X_d)$.

Exchangeability is a common a priori assumption in Bayesian statistics~\cite{definetti}, permutation hypothesis testing~\cite{good2002permutation}, and coalescent theory~\cite{kingman}. The class of exchangeable models contains all finite sequences of independent and identically distributed (i.i.d.) random variables, but is generally much larger. In fact, sampling colored balls from an urn without replacement produces a random sequence of colors which is exchangeable but not necessarily independent.

In what follows, $(X_1,\ldots,X_d)$, with $d>1$ finite, is a random vector with probability distribution $P$. Each $X_i$ is assumed to take values in a certain finite set $\curlyX$ of cardinality $k>1$. In addition, $\curlyP$ denotes the set of all probability models over $\curlyX^d$, and $\curlyE\subset\curlyP$ the subset of exchangeable models, which is clearly closed. We refer to the latent weight of $P$ with respect to $\curlyE$, $\lambda_\curlyE(P)$, as the exchangeable weight of $P$.

\begin{definition}
For each $x\in\curlyX^d$, let $[x]\subset\curlyX^d$ denote the set of all vectors of the form $(x_{\sigma(1)},\ldots,x_{\sigma(d)})$, with $\sigma$ a permutation of $(1,\ldots,d)$.
\label{def:equivclass}
\end{definition}

The set $\curlyX^d$ has $\binom{k+d-1}{d}$ permutation-equivalence classes. (The number of permutation equivalence classes equals the number of ways to place $d$ unlabelled balls in $k$ labelled urns.) In fact, $\curlyE$ is a simplex, in particular, also a convex set, with $\binom{k+d-1}{d}$ extreme points, each of which is a probability model having uniform mass over a single permutation-equivalence class~\cite{diaconis1978}.

In what follows, $[\curlyX^d]$ denotes the set of permutation equivalence classes of $\curlyX^d$. Next we show various properties of exchangeable weights, starting with the following explicit formula.

\begin{theorem}
\label{thm:3}
For all $P\in\curlyP$:
\begin{align} 
\lambda_\curlyE(P) =\sum_{x\in\curlyX^d} \min_{y\in [x]} P(y) =  \sum_{z\in[\curlyX^d]} |z|\cdot\min_{x\in z}P(x).
\label{eq:lambda}
\end{align}
If $\lambda_\curlyE(P)>0$ then exactly one probability measure $Q\in\curlyE$ is associated with the exchangeable weight of $P$: for each $x\in\curlyX^d$, $Q(x)=\min_{y\in[x]}P(y)/\lambda_\curlyE(P)$.
\end{theorem}

\begin{proof}
Define $\lambda^*:=\sum_{x\in\curlyX^d} \min_{y\in [x]} P(y)$. Suppose $\beta\ge0$ and $Q'\in\curlyE$ are such that $P\geq\beta\cdot Q'$. Since exchangeability implies $Q'(y)=Q'(x)$ for all $y\in[x]$, it follows that 
\[\beta\cdot Q'(x) \leq \min_{y\in[x]}P(y).\]
Moreover, because $Q'$ is a probability measure, we have that 
\[\beta = \sum_{x\in\curlyX^d}\beta\cdot Q'(x) \leq  \sum_{x\in\curlyX^d} \min_{y\in[x]}P(y)=\lambda^*,\]
which implies that $\lambda_\curlyE(P)\le\lambda^*$. But observe that $P\ge\lambda^*\cdot Q$, where $Q(x):=\min_{y\in[x]}P(y)/\lambda^*$. Since $Q$ is an exchangeable probability measure over $\curlyX^d$, it follows that $\lambda_\curlyE(P)\ge\lambda^*$, hence $\lambda_\curlyE(P)=\lambda^*$. The second identity in equation~(\ref{eq:lambda}) is now direct from this equality.

Finally, assume that $\lambda^*>0$ and suppose that $S\in\curlyE$ is such that $P\ge\lambda^* \cdot S$. We show using an argument by contradiction that $S\le Q$. Indeed, if there were $x\in\curlyX^d$ such that $S(x)>Q(x)$ then, because $Q$ and $S$ are exchangeable, the definition of $Q$ would imply that there is $y\in[x]$ such that $S(y)>Q(y)$ and $\lambda^*\cdot Q(y)=P(y)$. In particular, $P(y)<\lambda^* \cdot S(y)$, which is not possible. Thus $S\le Q$, which implies that $S=Q$ as claimed.
\end{proof}

Due to Theorem~\ref{thm:supminformula}, whenever $P$ is not itself exchangeable i.e. $\lambda_\curlyE(P)<1$, $P$ also has a unique component $R:= (P-\lambda\cdot Q)/(1-\lambda)$, which we call the unexchangeable component, which admits no exchangeable component of its own. In this sense $P$ can be distilled entirely into its exchangeable and unexchangeable parts. This property allows one to combine an exchangeable probability model with one that is totally unexchangeable so that the resultant source has a desired exchangeable weight:

\begin{corollary}
If $P=\beta\cdot Q' + (1-\beta)\cdot R'$, where $0\le\beta\le1$, $Q'\in\curlyE$, and $R'\in\curlyP\setminus\curlyE$ is such that $\lambda_\curlyE(R')=0$, then $\lambda_\curlyE(P)=\beta$.
\end{corollary} 
\begin{proof}
Since $\lambda_\curlyE(R')=0$, the identity in equation~(\ref{eq:lambda}) implies that for each $z\in[\curlyX^d]$ there is $y\in z$ such that $R'(y)=0$. In particular, because $Q'$ is exchangeable, reusing equation~(\ref{eq:lambda}) we obtain that:
\[\lambda_\curlyE(P)=\sum_{x\in\curlyX^d}\left(\beta\cdot Q'(x) + (1-\beta)\cdot \min_{y\in[x]}R'(y)\right)=\beta\cdot \sum_{x\in\curlyX^d}Q'(x)=\beta.\]
\end{proof}

\subsection{Bounds on exchangeable weights}\label{sec:upperbounds}

In what follows, for each non-empty $I\subset\{1,\ldots,d\}$, $P_I$ denotes the marginal distribution of $(X_i)_{i\in I}$. Clearly, $P_I$ is exchangeable when $P$ is exchangeable.

The following result may be useful to estimate an upper-bound on the exchangeable weight of a probabilistic source on $\curlyX^d$. Indeed, estimating $\lambda_\curlyE(P)$ when $P$ is known only indirectly through data requires estimating $(k^d-1)$ free parameters, which may be infeasible in practice. However, marginalizing $X$ to a random sub-vector of dimension $s<d$, one may significantly reduce the dimension of the estimation problem.

\begin{theorem}
The exchangeable weight of a full $d$-dimensional joint distribution is a lower bound on the exchangeable weight of any marginal, i.e. $\lambda_{\curlyE}(P)\leq\lambda_{\curlyE}(P_I)$, for all $P\in\curlyP$ and non-empty $I\subset\{1,\ldots,d\}$. 
\label{thm:marginalUB}
\end{theorem}

In particular, if $\lambda_\curlyE(P_I)$ is small for some $I$ then so is $\lambda_\curlyE(P)$; in which case, only very few of the data produced by $P$ could be attributed to an exchangeable source.

\begin{proof}
For each $x\in \curlyX^{|I|}$ and  $\alpha\in\curlyX^{d-|I|}$, let $x\alpha$ be the vector $y\in\curlyX^d$ such that $(y_i)_{i\in I}=x$ and $(y_i)_{i\notin I}=\alpha$. If $Q$ is the exchangeable probability measure given in equation~(\ref{eq:lambda}) then
\[P_I(x)=\sum_{\alpha\in\curlyX^{d-|I|}}P(x\alpha)\geq \sum_{\alpha\in\curlyX^{d-|I|}}\lambda_{\curlyE}(P)\cdot Q(x\alpha)= \lambda_{\curlyE}(P)\cdot Q_I(x),\]
for each $x\in\curlyX^{|I|}$. Hence, since $Q_I$ is exchangeable, $\lambda_{\curlyE}(P)\leq\lambda_{\curlyE}(P_I)$.
\end{proof}

The tightness of the inequality in the theorem is related to the notion of extendibility~\cite{diaconis1980}. Let $1\le s\le d$. An (exchangeable) probability model $\mu$ on $\curlyX^s$ is called $d$-extendible if there is an exchangeable probability measure $\nu$ on $\curlyX^d$ and $I\subset\{1,\ldots,d\}$ of cardinality  $s$ such that $\mu=\nu_I$. Necessary and sufficient conditions for extendibility can be found in~\cite{extendibility1}.

\begin{corollary}
For any $P\in\curlyP$ with $\lambda_\curlyE(P)>0$, and $I\subset\{1,\ldots,d\}$, if $\lambda_\curlyE(P)=\lambda_\curlyE(P_I)$ then the exchangeable component of $P_I$ is $d$-extendable.
\end{corollary}

\begin{proof}
Suppose that $\lambda_\curlyE(P)=\lambda_\curlyE(P_I)>0$, with $|I|=s$. Let $Q$ and $\tilde Q$ denote the  exchangeable components of $P$ and $P_I$, respectively. From the proof of Theorem~\ref{thm:marginalUB}, $P_I\ge\lambda_\curlyE(P)\cdot Q_I=\lambda_\curlyE(P_I)\cdot Q_I$; in particular, due to Theorem~\ref{thm:3}, $\tilde Q=Q_I$. 
\end{proof}

The converse in the corollary is not necessarily true, however. For a counterexample, consider $P\in\curlyP(\{0,1\}^3)$ such that $P(1,0,1)=P(1,1,0)=P(1,1,1)=1/3$; in particular, $\lambda(P)=1/3$. If $I=\{1,2\}$ then $P_I(1,0)=1/3$ and $P_I(1,1)=2/3$, hence $\lambda(P_I)=2/3$, and $\tilde Q=\delta_{(1,1)}$, which is $3$-extendible to $\delta_{(1,1,1)}$ yet $\lambda(P)\ne\lambda(P_I)$.

Finally, another way to bound the exchangeable weight of a probability measure on $\curlyX^d$ is to lump states in $\curlyX$ as follows.

\begin{theorem}
Let $\curlyY$ be a finite set and $r:\curlyX\to\curlyY$ a function, and define $\Phi:\curlyX^d\to\mathcal{Y}^d$ as $\Phi(x):=\big(r(x_1),\ldots,r(x_d)\big)$. Then, for each $P\in\curlyP(\curlyX^d)$, $\lambda_{\curlyE}(P)\leq \lambda_{\curlyE}(P\circ\Phi^{-1})$, where $P\circ\Phi^{-1}$ is the forward  measure of $P$ by $\Phi$.
\label{thm:lumping}
\end{theorem}

\begin{proof}
If $Q$ denotes the exchangeable component of $P$ then
\begin{align*}
\lambda_{\curlyE}(P\circ\Phi^{-1}) &= \sum_{[y]\subset\mathcal{Y}^d} |[y]|\cdot\min_{z\in[y]} P(\Phi^{-1}(z)) \\
&\geq \sum_{[y]\subset\mathcal{Y}^d} |[y]|\cdot\min_{z\in[y]} \lambda_{\curlyE}(P)\cdot Q(\Phi^{-1}(z)) \\
&= \lambda_{\curlyE}(P)\cdot\sum_{[y]\subset\mathcal{Y}^d} |[y]|\cdot\min_{z\in[y]} Q(\Phi^{-1}(z))=\lambda_{\curlyE}(P),
\end{align*}
where for the very last identity we have used that $Q$ is exchangeable; in particular, $\min_{z\in[y]}Q(\Phi^{-1}(z)) = Q(\Phi^{-1}(y))$, for each $y\in\mathcal{Y}^d$.
\end{proof}

\section{Estimation of exchangeable weights}\label{sec:estimation}

In this section, $X_1,\ldots,X_n$ denote $d$-dimensional i.i.d. samples from a probability measure $P$ defined over $\curlyX^d$. Define $\lambda:=\lambda_\curlyE(P)$.

Let $\hat P_n:=\sum_{i=1}^n \delta_{X_i}/n$ denote the empirical measure associated with the sample. A natural estimator of $\lambda$ is $\hat{\lambda}_n := \lambda_\curlyE(\hat{P}_n)$. Since $\hat P_n$ is a maximum likelihood estimator (MLE) of $P$, $\hat{\lambda}_n$ is a MLE of $\lambda$. Moreover, as noted in the proof of Theorem~\ref{thm:supminformula}, $\lambda_\curlyE(\cdot)$ is continuous. In particular, since $\hat P_n\to P$ almost surely, $\hat\lambda_n\to\lambda$ also almost surely. Since $0\le\hat\lambda_n\le1$ for all $n\ge1$, it follows that $\hat\lambda_n$ is an asymptotically unbiased estimator of $\lambda$. Nevertheless, because for each $z\in[\curlyX^d]$ the transformation $P\to\min_{y\in z}P(y)$ is concave down, equation~(\ref{eq:lambda}) implies that $\lambda_\curlyE(\cdot)$ is concave down. Thus, by Jensen's inequality, $E\big(\hat\lambda_n\big)\leq\lambda$, i.e. $\hat\lambda_n$ is a negatively biased estimator of $\lambda$.

For each $x\in\curlyX^d$, consider the quantities
\[m_x:=\min_{y\in [x]}P(y);\qquad\sigma_x^2 := m_x(1-m_x).
\]
Since these quantities remain constant within each permutation equivalence class, we sometimes abuse the notation and write for $z\in[\curlyX^d]$: $m_z$ and $\sigma_z^2$ to mean $m_x$ and $\sigma_x^2$, with $x\in z$, respectively.

Our next result characterizes implicitly the asymptotic distribution of $\hat\lambda_n$.

\begin{theorem}
If for each $z\in[\curlyX^d]$, $C_z:=\{x\in z\hbox{ such that }P(x)=m_x\}$ and $c_z:=|C_z|$ then
\begin{equation}
\lim_{n\to\infty}\sqrt{n}\big(\hat\lambda_n-\lambda\big)\stackrel{d}{=}\sum_{z\in[\curlyX^d]}|z|\cdot \min_{x\in C_z}Z_z^{(x)},
\label{ide:intermediate}
\end{equation}
where $(Z_z^{(x)})_{z\in[\curlyX^d],x\in C_z}$ is a normal random vector such that, for each $z\in[\curlyX^d]$, $\big(Z_z^{(x)}\big)_{x\in C_z}$ is a $c_z$-dimensional zero-mean exchangeable normal random vector with variance-covariance matrix $\Sigma_z$ such that $\Sigma_z(x,x)=\sigma_z^2$ and $\Sigma_z(x,y)=-m_z^2$, for all $x,y\in C_z$ with $x\ne y$, and for $z_1,z_2\in[\curlyX^d]$ with $z_1\ne z_2$, $\hbox{cov}\big(Z_{z_1}^{(x)},Z_{z_2}^{(y)}\big)=-m_{z_1}m_{z_2}$, for all $x\in z_1$ and $y\in z_2$. 

\label{thm:asympexchangeable}
\end{theorem}

\begin{proof}
Our arguments follow closely those in~\cite{bootstrapmax}.

For the sake of notation, we remove the sub-index $n$ from quantities defined in terms of $\hat P_n$; in particular, we write $\hat P$ instead of $\hat P_n$ and $\hat\lambda$ instead of $\hat\lambda_n$. For $x\in\curlyX^d$, define $\hat M_x:=\min_{y\in C_{[x]}}\hat P(y)$ and $\hat m_x :=\min_{y\in[x]}\hat P(y)$. According to the Law of Large Numbers, $\hat P\to P$ almost surely; in particular:
\begin{align}
\Prob\Big(\hbox{there is $x\in\curlyX^d$ such that }\hat M_x\ne\hat m_x\Big)=o(1).
\label{eq:minland}
\end{align}

Define $\delta:=\sqrt{n}\big(\hat\lambda-\lambda\big)$ and $\Delta:=\sqrt{n}\big(\sum_{[x]\subset\curlyX^d}|[x]|\cdot\min_{y\in C_x}\hat P(y)-\lambda\big)$. That is, $\Delta$ is proportional to the difference between $\lambda$ and an estimator of it computed from only those $\hat P(x)$ corresponding to outcomes $x\in\curlyX^d$ for which $P(x)=m_x$. 

Fix $t\in\mathbb{R}$. By equations~(\ref{eq:lambda}) and~(\ref{eq:minland}), 
\begin{equation}
|\Prob(\delta\leq t)-\Prob(\Delta\leq t)|=o(1).
\label{eq:lamsub}
\end{equation}
Furthermore, due to the well-known Central Limit Theorem for the multinomial distribution:
\begin{align}
\Prob(\Delta\leq t) &=\Prob\Big(\sum_{[x]\subset\curlyX^d}|[x]|\cdot\min_{y\in C_x} \sqrt{n}\cdot\big(\hat P(y)-m_x\big)\leq t\Big)\nonumber\\
&\qquad\longrightarrow \Prob\Big(\sum_{[x]\subset\curlyX^d}|[x]|\cdot\min_{y\in C_{[x]}}Z_{[x]}^{(y)}\leq t\Big),
\label{eq:MLEconv}
\end{align}
where $(Z_z^{(x)})_{z\in[\curlyX^d],x\in C_z}$ is a zero-mean normal random vector with variance-covariance matrix as described above. The theorem is now a direct consequence of equations~(\ref{eq:lamsub}) and (\ref{eq:MLEconv}).
\end{proof}

Let $\curlyP_U\subset\curlyP$ denote the set of $P\in\curlyP$ such that $c_x=1$ for each $x\in\curlyX^d$, i.e. for each $z\in[\curlyX^d]$ there is a \underline{u}nique $y\in z$ such that $P(y)=m_z$. The following result is an almost direct consequence of the previous theorem. This result also follows, albeit less directly, from the multivariate delta method~\cite{vandervaart}.

\begin{corollary}
If $P\in\curlyP_U$ then
\[\lim_{n\to\infty}\sqrt{n}\big(\hat\lambda-\lambda)\stackrel{d}{=}Z,\]
where $Z$ is a zero-mean normal random variable with variance
\[V(Z)=\sum_{[x]\subset\curlyX^d}|[x]|^2\,\sigma_x^2-\sum_{[x]\ne[y]}|[x]|\,|[y]|\,m_xm_y.\]
\label{cor:P1limit}
\end{corollary}

\begin{proof}
Due to the hypothesis on $P$,
\[\sqrt{n}(\hat\lambda-\lambda)\stackrel{d}{\longrightarrow}Z=\sum_{[x]\subset\curlyX^d}|[x]|\cdot Z_{[x]},\]
where $\big(Z_{[x]}\big)_{[x]\subset\curlyX^d}$ is an ${k+d-1\choose d}$-dimensional exchangeable normal random vector such that $E\big(Z_{[x]}\big)=0$, $V\big(Z_{[x]}\big)=m_x(1-m_x)$, and for $[x]\ne[y]$, $\hbox{cov}\big(Z_{[x]},Z_{[y]}\big)=-m_xm_y$. In particular, $Z$ has a normal distribution, from which the corollary follows.
\end{proof}

\color{black}

Let $\mathbb{X}:=(X_1,\ldots,X_n)$ denote an i.i.d. sample from $P\in\curlyP$, $\mathbb{X}^{*}:=(X_{1}^{*},\ldots,X_{n}^{*})$ denote a single resample with replacement from $\mathbb X$, and $\hat\lambda^*$ denote the exchangeable weight associated with the empirical measure $\hat P^*:=\sum_{i=1}^n\delta_{X_i^*}/n$. In the next theorem and corollary we characterize the asymptotic distribution of the bootstrap distribution estimator of $\sqrt{n}(\hat\lambda^*-\hat\lambda)$. In what follows, $Z=(Z_z^{(x)})_{z\in[\curlyX^d],x\in C_z}$ denotes the normal random vector described in Theorem~\ref{thm:asympexchangeable}, which has dimension $\sum_{z\in[\curlyX^d]}c_z$.

\begin{theorem}
Fix $t\in\mathbb R$, and associate with each vector $v:=(v_z^{(x)})_{z\in[\curlyX^d],x\in C_z}$, the function:
\[\psi(v) := \Prob\left(\sum_{z\subset[\curlyX^d]}|z|\cdot\Big[\min_{x\in C_z}\big(Z_z^{(x)}+v_z^{(x)}\big)-\min_{x\in C_z}v_z^{(x)}\Big]\leq t\right).\]
Note that $\psi(0)=\lim_{n\to\infty} \Prob\big(\sqrt{n}(\hat\lambda-\lambda)\leq t\big)$, the C.D.F. of the limiting random variable described in Theorem~\ref{thm:asympexchangeable}. The bootstrap estimator of $\psi(0)$ is
\[\hat\psi:=\Prob\Big(\sqrt{n}(\hat\lambda^*-\hat\lambda)\leq t\mid\mathbb{X}\Big),\]
and
\begin{align}
\hat\psi\overset{d}{\to}\psi(Y),
\end{align}
where $Y$ is an independent copy of $Z$. 
\label{thm:bootstrap}
\end{theorem}

The proof of this theorem resembles closely the arguments given in~\cite{bootstrapmax}. We first  require the following lemma.

\begin{lemma}
Define $\hat M^*_x:=\min_{y\in C_x}\hat P^*(y)$ and $\hat m^*_x:=\min_{y\in[x]}\hat P^*(y)$, i.e. $\hat m^*_x$ is the minimum probability estimated from a bootstrap resample in each equivalence class, and $\hat M^*_x$ is the same, estimated only from $C_x$. Then \[\Prob\big(\text{there is }x\in\curlyX^d\text{ such that }\hat M^*_x\neq \hat m^*_x\mid\mathbb X\big)=o_p(1).\]\label{lemma:Cxpbound}
\end{lemma}

\begin{proof}
If $P\in\curlyE$, $\hat M_x=\hat m_x$, and $\hat M^*_x=\hat m_x$ with probability $1$ for each $x$, and the claim follows.

Instead, for $P\in\curlyP\setminus\curlyE$, define \[\xi:=\min_{[x]\subset\curlyX^d\text{ s.t. }|[x]|<c_x}\left\{\min_{y\in [x]\setminus C_x} p_y-m_x\right\},\] 
that is, the smallest difference between $p_x$, for $x\not\in C_x$, and $m_x$. Note that $\xi$ is positive. Similarly, define $\hat\xi:=\min_{[x]\subset\curlyX^d\text{ s.t. }|[x]|<c_x}\{\min_{y\in [x]\setminus C_x} \hat p_y-\hat M_x\}$. (If $\hat\xi<0$, then $\hat M_x>\hat m_x$ for some $[x]$.) The Law of Large Numbers guarantees that $\Prob(\hat\xi\geq\xi/2)=(1-o(1))$.

Define $Y_n:=\Prob\big(\text{there is }x\in\curlyX^d\text{ such that }\hat M^*_x\neq \hat m^*_x\mid\mathbb X\big)$, a random variable taking values in $[0,1]$. Fix $\epsilon>0$. Then
\begin{align*}
    \Prob(Y_n\geq \epsilon)&=\Prob(Y_n\ge\epsilon\mid\hat\xi\geq\xi/2)\cdot (1-o(1)) +\Prob(Y_n\ge\epsilon\mid\hat \xi<\xi/2)\cdot o(1).
\end{align*}
A second application of the Law of Large Numbers now yields $\Prob(Y_n\ge\epsilon\mid\hat\xi\geq\xi/2)\to0$, proving the lemma.


\end{proof}

\noindent\textit{Proof of Theorem~\ref{thm:bootstrap}.} Define
\begin{align*}
\delta^* &:=\sqrt{n}\big(\hat\lambda^*-\hat\lambda\big)\\   
\Delta^* &:=\sqrt{n}\big(\sum_{z\in[\curlyX^d]}|z|\cdot(\hat M^*_z-\hat M_z)\big).
\end{align*}
Due to Lemma~\ref{lemma:Cxpbound}:
\begin{align}\Prob(\delta^*\leq t\mid\mathbb X)=\Prob(\Delta^*\leq t\mid\mathbb X)+o_p(1).\label{eq:BSconv}
\end{align}
Then it follows:
\begin{align}
\Prob(\Delta^*\leq t\mid\mathbb X)
&= \Prob\left(\sqrt{n}\sum_{z\in[\curlyX^d]}|z|\cdot\Big(\hat M^*_z-\hat M_z\Big)\leq t|\mathbb X\right)\nonumber\\
&= \Prob\Bigg(\sum_{z\in[\curlyX^d]}|z|\cdot\Big[\min_{x\in C_z}\big\{\sqrt{n}(\hat p^*_{x}-\hat p_{x})+\sqrt{n}(\hat p_{x}-m_z)\big\}\label{eq:BSfinite}\\
&\qquad\qquad\qquad\qquad\qquad\qquad-\min_{x\in C_z}\sqrt{n}(\hat p_{x}-m_z)\Big]\leq t|\mathbb X \Bigg) \nonumber \\
&= \psi\Big(\big(\sqrt{n}[\hat p_x-m_z]  \big)_{z\in[\curlyX^d],x\in C_z}\Big)+o_p(1).
\label{eq:weakBSconv}
\end{align}
The last equality is due to the fact that, for almost every sample sequence $X_1,X_2,\ldots$, $\sqrt{n}(\hat P^*-\hat P)$ has the same conditional limiting distribution given $\mathbb{X}$ that $\sqrt{n}(\hat P-P)$ does unconditionally, because $\hat P$ has finite second moments and mean $P$~\cite[Theorem 2.2]{bickelfreedman1981}. That is, for any $s\in\mathbb{R}^{k^d}$, $\Prob(\sqrt{n}(\hat P^*-\hat P)\leq s\mid \mathbb X)\to \lim_{n\to\infty}\Prob(\sqrt{n}(\hat P-P)\leq s)$ almost surely (and therefore in probability). For this reason, we can replace $\sqrt{n}(\hat P^*-\hat P)\mid \mathbb X$ in line~(\ref{eq:BSfinite}) by its limiting random vector, which introduces some random perturbation to the conditional C.D.F. which converges in probability to $0$. 

Finally, because $\psi(v)$ is a continuous function of $v$, we take equation~(\ref{eq:weakBSconv}) with equation~(\ref{eq:BSconv}) to find that \[\Prob(\delta^*\leq t\mid\mathbb X)=\psi\Big(\big(\sqrt{n}[\hat p_x-m_z]\big)_{z\in[\curlyX^d],x\in C_z}\Big)+o_p(1)\overset{d}{\to}\psi\big((Y^{(x)}_z)_{z\in[\curlyX^d],x\in C_z}\big),\] 
as claimed.

\hfill$\Box$

The above convergence is weak, however, by requiring $P\in\curlyP_U$, we can ensure convergence in probability of the bootstrap distribution estimator. The following corollary follows simply from the fact that when $c_z=1$ for each $z\in[\curlyX^d]$, $\min_{x\in C_z}(Z_z^{(x)}-Y_z^{(x)})-\min_{x\in C_z}Y_z^{(x)} = Z_z^{(x)}$.

\begin{corollary}
When $P\in\curlyP_U$,
\begin{align*}
    \Prob\Big(\sqrt{n}(\hat\lambda^*-\hat\lambda)\leq t\mid\mathbb X\Big)\overset{p}{\to}\Prob\Big(\sqrt{n}(\hat\lambda-\lambda)\leq t\Big),
\end{align*}
for each $t\in\mathbb R$.
\end{corollary}

Let $\curlyP_N:=\curlyP\setminus\curlyP_U$ denote the set of sources which have at least one permutation class $z\in[\curlyX^d]$ where $c_z>1$, i.e., sources for which at least one $m_z$ is \underline{n}ot uniquely achieved by $y\in z$. When $P\in\curlyP_N$, a more explicit characterization of the distribution of $Z$ seems very elusive, and in particular, is not Gaussian. To see why, observe that \cite[Corollary 5]{maxEC} implies that $\min_{x\in C_z}Z_z^{(x)}$ has probability density function (p.d.f.):
\[f_z(t)=\frac{c_z}{\sigma_z}\cdot\varphi\left(\frac{t}{\sigma_z}\right)\cdot\Phi_{c_z-1}\left(\frac{t\sqrt{m_z}}{\sigma_z^2},\ldots,\frac{t\sqrt{m_z}}{\sigma_z^2}\,;\,\rho_x\mathbb{I}_{c_z-1}+\frac{\Sigma_x}{\sigma_z^2}\right),\]
where $\varphi(\cdot)$ is the p.d.f. of a standard Normal random variable, $\Phi_k(\,\cdot\,;\Sigma)$ is the cumulative distribution function (c.d.f.) of a zero-mean $k$-dimensional multivariate normal distribution with variance-covariance matrix $\Sigma$, and $\mathbb{I}_k$ is the $k$-dimensional identity matrix.

Unfortunately, the above probability densities are not enough to describe the distribution of $Z$ due to the correlation between the minima in equation~(\ref{ide:intermediate}). Furthermore, Theorem~\ref{thm:bootstrap} does not hold when $P\in\curlyP_N$, because whenever $c_x>1$, $\min_{x\in C_z}(Z_z^{(x)}+Y_z^{(x)})-\min_{x\in C_z}Y_z^{(x)}\neq\min_{x\in C_z}Z_z^{(x)}$ with probability $1$. As a result, the weak convergence in the last line of equation~(\ref{eq:weakBSconv}) is to a version of the distribution of $Z$, but with some Gaussian perturbation. 

Luckily, however, the bootstrap estimator of $Z$ when $P\in\curlyP_U$ can be made consistent by choosing a resample size $n_0$ of order $o(n)$. In this case, we would redefine $\delta^*:=\sqrt{n_0}\big(\hat\lambda^*_n-\hat\lambda\big)$ and $\Delta^*:=\sqrt{n_0}\big(\sum_{z\in[\curlyX^d]}|z|\cdot(\hat M^*_{z}-\hat M_z)\big)$, so that
\begin{align*}
    \hat\psi &= \Prob(\delta^*\leq t\mid\mathbb X) = \Prob(\Delta^*\leq t\mid \mathbb X)+o_p(1)\nonumber\\
    &= \psi\Big(\big(\sqrt{n_0}(\hat p_x-p_x)_{z\in[\curlyX^d],x\in C_z}  \big)\Big)+o_p(1)\\
    &\overset{p}{\to}\psi(0),
\end{align*}
because $\sqrt{n_0}(\hat p_{x}-p_{x})\overset{p}{\to} 0$.

Based on simulations, we have found that it is usually more accurate to use a full size-$n$ resample for the purposes of correcting the bias of $\hat\lambda_n$, even at moderate sample sizes and when $P\in\curlyP_N$. Additionally, Monte Carlo bootstrap estimates of $V(\hat\lambda)$ tend to be more accurate than the asymptotic formula given in derived from Corollary~\ref{cor:P1limit}. With very large samples, users may wish to try using a size $n_0:=2\sqrt{n}$ resample. Users might also explore ad hoc methods for combined estimators of $\min_{y\in[x]}p_y$ when there is strong reason to believe that $P\in\curlyP_N$. In what follows in this manuscript, bias and variance of $\hat\lambda$ are approximated from a Monte Carlo estimate of the bootstrap distribution $\sqrt{n}(\hat\lambda^*_n-\hat\lambda)\mid\mathbb X\overset{d}{\approx}\sqrt{n}(\hat\lambda-\lambda)$.

An explicit Berry-Esseen type bound on the error using the limiting normal distribution to approximate the sampling distribution of $\hat\lambda_n$ remains elusive, therefore we recommend selecting sample size by simulation. The largest difficulty in estimating $\hat\lambda$ is controlling negative bias, especially when $P\in\curlyP_N$. Therefore we recommend to simulate data from several test sources $T\in\curlyE$ to approximate a worst-case sampling distribution for a given sample size. In particular, it is suitable to choose $T$ with $T(\{a\}^d)=0$ for each $a\in\curlyX$, and uniform mass elsewhere, because each of the outcomes $\{a\}^d$ belongs to a singleton equivalence class. Any observation of these outcomes can only increase the estimate of $\lambda(T)$. In our experience on a variety of sample spaces $\curlyX^d$, estimation of $\lambda(T)$ has the largest bias and standard deviation of any source in $\curlyP(\curlyX^d)$.

We suggest the following heuristic: first, select several candidate sample sizes $(n_i)_{i=1,2,\ldots}$. For each candidate sample size, repeatedly simulate $n_i$ outcomes from the test source $T$ described above to get an empirical estimate of standard error and bias of $\lambda(\hat T)$ (the test source above satisfies $\lambda(T)=1$). After selecting a sample size for which standard deviation and bias of $\hat\lambda_{n_i}(T)$ appear acceptably small, collect samples of this size for every source in a coarse grid over $\curlyP(\curlyX^d)$. This is to ensure that even on the most pathological sources we can obtain acceptable estimates of $\lambda$.

\section{DNA methylation analysis}\label{sec:proofofconcept}


When a DNA sequence contains a cytosine residue (C) followed by a guanine residue (G) in the $5'$-to-$3'$ sense, this dimer is referred to as a CpG. The cytosine in a CpG may or may not have methyl group bonded to it in the $5'$ position of its pyrimidine ring. This methylation is regulated by reversible enzymatic processes and is known to modulate gene expression; increased methylation in gene promoters is associated with transcriptional silencing~\cite{methylationtranscription}, and specific DNA methylation patterns have been linked to human disease \cite{methylationDisease}. In particular, certain aberrant methylation patterns are a hallmark of some cancers \cite{methReview}, and as such considerable effort has been expended to determine regions of DNA that have differential methylation under different cellular conditions. 

One popular modern assay to assess DNA methylation is Whole-Genome Bisulfite Sequencing (WGBS)~\cite{WGBS}, a procedure in which unmethylated cytosines are chemically transformed into thymine (T) through treatment with a bisulfite catalyst. When bisulfite-treated DNA is then sequenced by high-throughput shotgun technology, methylated CpGs can be distinguished from unmethylated ones by the observation of a ``CG'' dimer versus a ``TG'' dimer, as depicted in Figure~\ref{fig:MethDiagram}.

\begin{figure}
\centering
\includegraphics[scale=0.25]{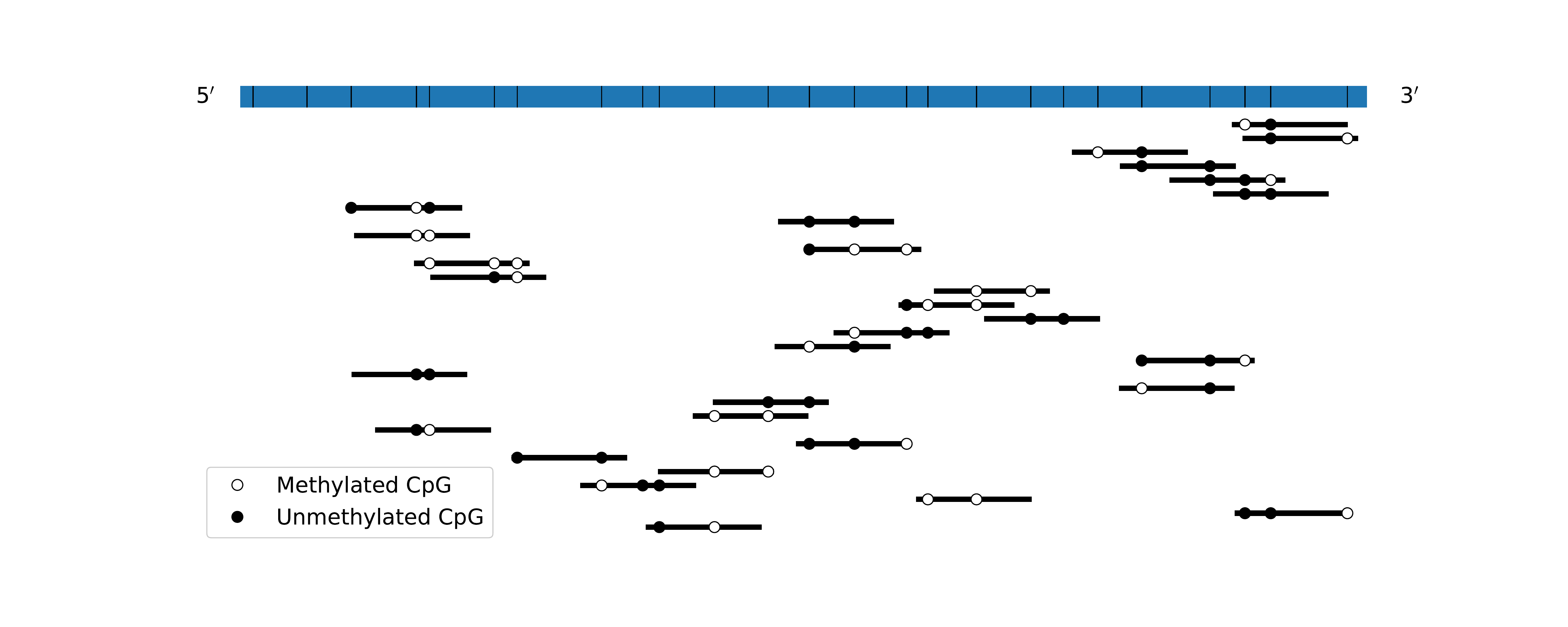}
\caption{Diagram of a typical WGBS experiment. The blue rectangle represents a segment of ssDNA, with the location of CpGs on that strand marked by black vertical bars. Black horizontal line segments represent reads mapped to a reference strand, open and closed circles represent the partially-observed joint methylation status several CpGs. }
\label{fig:MethDiagram}
\end{figure}

When attempting to describe DNA methylation, it is routine to use a sliding window approach~\cite{BSmooth,methylKit} wherein all observations of methylated and unmethylated CpGs are counted in the window, typically $1$ Kb in length, to summarize local methylation. It is common to compare methylation between two different biological samples using, e.g., Fisher's exact test~\cite{methylKit}, concluding that a window is differentially methylated if the null hypothesis of equal distribution can be rejected. This approach assumes that in a single window the methylation status of each CpG contributes identically in its biological effect. For example, if we denote an unmethylated CpG as a `$0$' and a methylated CpG as a `$1$' and consider ten consecutive CpGs in different tissues, with the first tissue always producing the configuration `$1111100000$' and the second always producing the configuration `$0000011111$', the above approach would be unable to identify this locus as differentially methylated. So current approaches for differential methylation implicitly assume that binary sequences representing methylation inside each window have an exchangeable distribution.

To evaluate this assumption, we examined WGBS data from 121 experimental replicates representing 77 unique biological samples, publicly available from ENCODE~\cite{encode,encodeUpdate}. These replicates include clinical tissue samples, cell lines, and primary cells. We selected replicates using single-end reads which were not flagged by ENCODE as having low coverage or insufficient read length. The list of the sample identifiers  (\texttt{ENCODE\_IDs.xlsx}) and processed datasets can be found on GitHub  (\url{https://github.com/antonypearson/OnContaminationofSymbolicDatasets}).

Each replicate is associated with a BAM file generated by mapping reads to GRCh38 using Bismark~\cite{bismark}. We used the MethPipe methylation software suite~\cite{methpipe} to convert BAM files into MethPipe format and generate epiread files, an efficient format reporting the genomic index and methylation status of each CpG contained in a read. 

To investigate the exchangeability of local DNA methylation we focused on sets of 3 consecutive CpGs, which we call ``triplets.'' Due to Theorem~\ref{thm:marginalUB}, a genomic region containing a highly unexchangeable triplet must have highly unexchangeable methylation overall. 

For each replicate, we used the epireads file generated to extract data from ``well-covered'' triplets---i.e. those where all three CpGs are jointly covered by at least $100$ reads, discarding reads which report a CpG with ambiguous methylation status. In all datasets, we observed $637,612$ well-covered triplets, representing $72,815$ unique loci. For each autosome in each sample we estimated the exchangeable weight of each well-covered triplet and corrected for estimator bias using a sample mean of $N=1000$ full bootstrap resamples. Although each estimate of a triplet's exchangeable weight $\hat\lambda$ lies in $[0,1]$, the bias-adjusted estimate $(\hat\lambda-\bar{\hat\lambda^*})$ may be larger than $1$ or smaller than $0$. Therefore we truncate these estimates to $[0,1]$. Available online are Numpy files containing processed triplets corresponding to each BAM file ID. Each row corresponds to a well-covered triplet, with columns corresponding to 1) chromosome number, 2) index of the triplet on the chromosome, 3) an estimate of the total variation distance to the class of exchangeable distributions, 4) an estimate of the exchangeable weight of the triplet (bias-corrected), 5) a bootstrap estimate of the standard deviation of $\hat\lambda$, 6-13) the counts of each of the $8$ possible triplet configurations (ordered lexicographically, i.e. `000', `001', etc.), and 14-21) an estimate of the largest exchangeable component. 

Estimates of triplet exchangeable weight are depicted in Figure~\ref{fig:chromExch}. As seen in the figure, in some chromosomal regions, particularly e.g. on chromosomes 6 and 13, there are triplets whose exchangeable weights are very small. In fact, some appear completely unexchangeable. 

\begin{figure}[h!]
\centering
\includegraphics[scale=0.6]{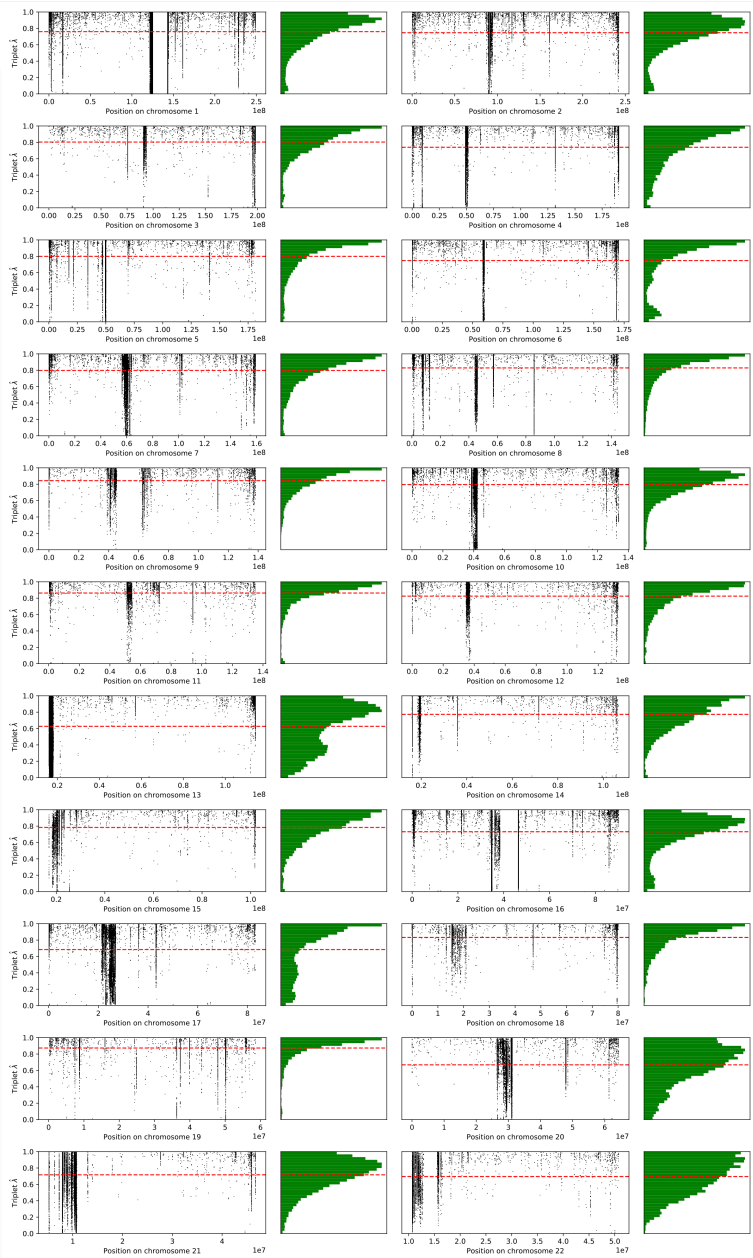}
\caption{Estimated exchangeable weights of well-covered triplets by chromosome. Dashed red lines denote the mean, and green plots the histograms associated with these weights.}
\label{fig:chromExch}
\end{figure}

As seem in Table~\ref{tab:TSSdistCorr}, triplet exchangeability does not appear strongly correlated with the genomic distance to the nearest promoter. Further, as seen in Figure~\ref{fig:TSSdistCorrFig}, within each chromosome, and within each dataset, the correlation between triplet exchangeability and distance from a promoter is usually small. There is a noticeable trend, however, that triplet exchangeabile weight is more likely to be negatively correlated with distance from a promoter. Indeed, both a two-sided Wald test and a Spearman rank-order test of the null hypothesis that TSS proximity and estimated triplet exchangeable weight are uncorrelated give very small p-values ($p\ll 10^{-10}$). That is, despite the small magnitude of the effect, we can detect that triplets close to promoters tend to have more-exchangeable methylation configurations. 

\begin{table}
\centering
\caption{Correlation per chromosome between estimated exchangeable weight of well-covered triplets, and distance between their center and the nearest transcription start site (TSS).}
\label{tab:TSSdistCorr}
\tiny
\begin{tabular}{cccc}
\toprule
Chromosome & Correlation & Chromosome & Correlation\\
\cmidrule(lr){1-2}
\cmidrule(lr){3-4}
1 & $-1.200\times10^{-1}$ & 12 & $+2.676\times10^{-2}$\\
2 & $+1.236\times10^{-1}$ & 13 & $-2.406\times10^{-1}$\\
3 & $+4.738\times10^{-2}$ & 14 & $-2.977\times10^{-1}$\\
4 & $-1.742\times10^{-1}$ & 15 & $-1.999\times10^{-1}$ \\
5 & $-1.584\times10^{-1}$ & 16 & $+4.356\times10^{-2}$ \\
6 & $+7.900\times10^{-2}$ & 17 & $+9.088\times10^{-2}$ \\
7 & $+7.423\times10^{-2}$ & 18 & $+8.553\times10^{-2}$ \\
8 & $+2.869\times10^{-2}$ & 19 & $-6.973\times10^{-2}$ \\
9 & $+2.586\times10^{-2}$ & 20 & $-8.210\times10^{-2}$ \\
10 & $-3.123\times10^{-2}$ & 21 & $+6.563\times10^{-2}$ \\
11 & $+5.987\times10^{-2}$ & 22 & $-2.552\times10^{-1}$ \\
\bottomrule
\end{tabular}
\end{table}

\begin{figure}[h!]
\centering
\includegraphics[scale=1.0]{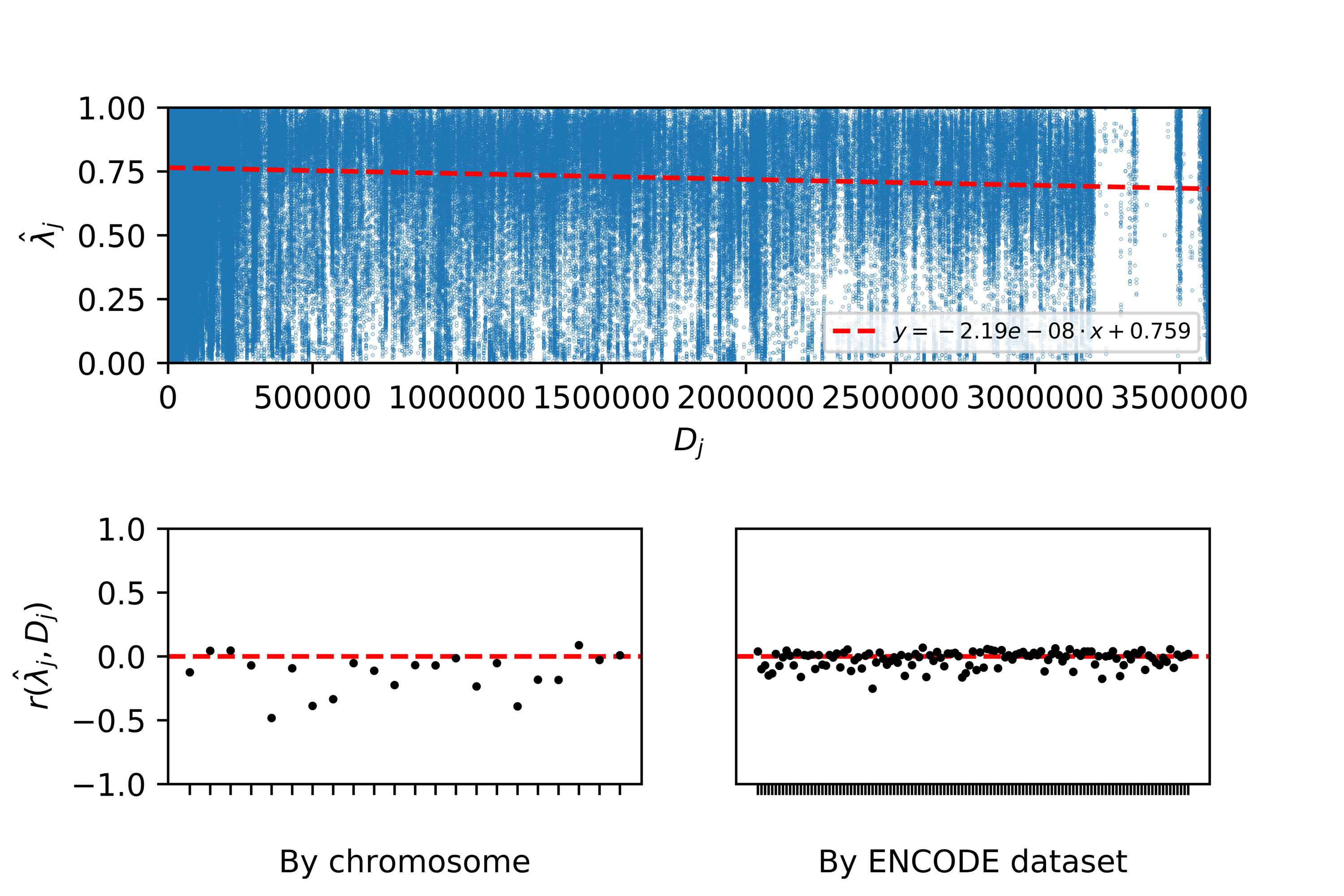}
\caption{Top, plot of estimated exchangeable weight for each well-covered triplet versus its distance from the nearest TSS. Bottom, correlation between distance to the nearest TSS and exchangeable weight of each triplet, per chromosome (left) and per dataset (right).
}
\label{fig:TSSdistCorrFig}
\end{figure}


\begin{figure}[h!]
\centering
\includegraphics[scale=0.7]{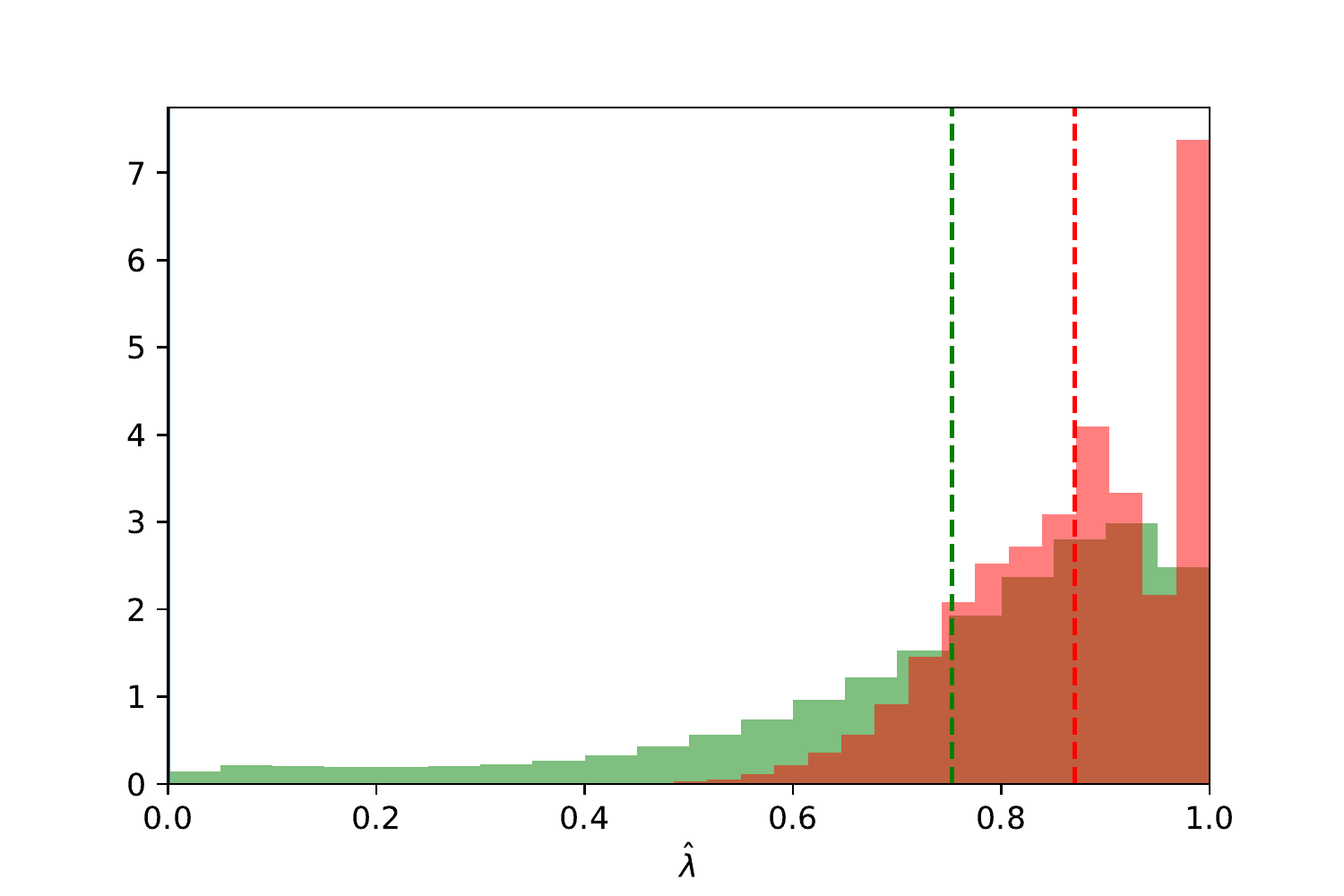} 
\caption{Histogram of estimated exchangeable weights of all triplets from aggregated WGBS data (green) and synthetic samples of size $n=100$ from the uniform distribution on $\{(0,0,1),(0,1,0),(0,1,1),(1,0,0),(1,0,1),(1,1,0)\}$ (red). The averages of estimated exchangeable weights from real and synthetic data are given by the dashed green and red lines, respectively.}
\label{fig:RealSynthetic}
\end{figure}

To confirm that these highly unexchangeable loci are not due to uncertainty in estimation of the exchangeable weight, we simulated data from the uniform distribution over  binary triplets except `000' or `111'. Based on empirical study, this source is the worst case for estimating the exchangeable weight in terms of bias and standard error. Nevertheless, as seen in Figure~\ref{fig:RealSynthetic}, the empirical distribution of estimated exchangeable weights of all triplets gives much greater probability mass near $0$ than the corresponding sampling distribution of the synthetic data. That is, uncertainty from statistical estimation does not account for the apparent phenomenon of highly unexchangeable loci. 

In general, it is impossible to disentangle contamination which is caused by, e.g., sequencing errors or incomplete enzymatic conversion of unmethylated cytosines, from biological processes discriminating specific configurations of methylation. Under the assumption that contamination of the former kind is small, i.e. that we have a truly accurate picture of how methylation is configured in cells, we would expect triplets to have exchangeable weights close to one if overall methylation levels govern biological function. This might mean that in some cell types methylation far away from promoters (and likely far from CpG islands) is ``locked in,'' and specific patterns of methylation rather than overall methylation levels modulate biological function.

Identifying the biological reason for highly unexchangeable loci remains an open question, which may not have a universal answer. We conclude that there are some loci which are far from exchangeable---that is, some configurations of methylation are discriminated at these triplets. The identification of these loci opens opportunities for more high-resolution understanding of methylation patterns. In particular, these loci represent regions where very specific configurations of methylation may regulate function.

\bibliographystyle{siamplain}

\end{document}